\documentclass[conference,10pt]{IEEEtran}
\usepackage[T1]{fontenc}
\usepackage{amsmath,amssymb,amsthm}
\usepackage{cite}
\usepackage[hidelinks]{hyperref}
\usepackage{xcolor}
\usepackage{comment}
\hypersetup{pdftitle={A pi/4 Lower Bound for Random Access in Symmetric Linear Ensembles}}

\newtheorem{theorem}{Theorem}
\newtheorem{lemma}[theorem]{Lemma}
\newtheorem{proposition}[theorem]{Proposition}
\newtheorem{corollary}[theorem]{Corollary}
\newtheorem{claim}[theorem]{Claim}
\theoremstyle{definition}

\DeclareMathOperator{\rank}{rank}
\newcommand{\spn}[1]{\langle #1\rangle}
\DeclareMathOperator*{\argmax}{arg\,max}
\DeclareMathOperator{\Pois}{Poisson}
\newcommand{\E}{\mathbb{E}}
\newcommand{\Prob}{\mathbb{P}}
\newcommand{\F}{\mathbb{F}}
\newcommand{\bfv}{\mathbf{v}}
\newcommand{\bfe}{\mathbf{e}}
\newcommand{\cA}{\mathcal{A}}
\newcommand{\cU}{\mathcal{U}}

\title{A $\pi/4$ lower bound for random access in symmetric linear ensembles under optimal decoding}
\author{\IEEEauthorblockN{Daniella Bar-Lev, Ohad Elishco, and Ryan Gabrys}}

\begin{document}
\maketitle

\begin{abstract}
We study the expected number of random linear observations required to recover a prescribed information symbol. For ensembles with uniformly random supports, independently sampled nonzero coefficients, and uniformly bounded column weights, we prove that the asymptotic normalized expected recovery time is at least $\pi/4$ under optimal decoding. The result holds over finite or infinite fields and extends the lower bound previously established for binary fully symmetric codes under peeling decoding. For singleton–pair observations, we determine the exact limiting expectation and show that optimal and peeling decoding have the same limit, independently of the field and coefficient distribution. The minimum within this family is approximately $0.793933445$, which explicit constructions attain asymptotically over every fixed finite field.
\end{abstract}


\section{Introduction}
The random-access problem in DNA-based storage asks how coding can reduce the number of random reads needed to recover a prescribed information symbol~\cite{cover}. In the noiseless uniform model, $k$ information symbols are encoded into $n$ symbols, which are sampled independently and uniformly with replacement. The objective is to minimize the expected number of samples needed for recovery, typically taking the maximum over the requested information symbols.

Related work considers several retrieval goals, including reconstructing the entire message~\cite{duality,optimal}, recovering requested files~\cite{files} or sets of information symbols~\cite{generalized}, characterizing tradeoffs between individual file retrieval times~\cite{block,hyperbolic}, and recovering a fraction of the message using LT codes~\cite{partial}. Other studies address errors and nonuniform sampling~\cite{hanna,cao,noisy}, repair of failed storage containers~\cite{repair}, and parallel sampling from multiple containers~\cite{distributed}. Further coverage models consider draws that reveal subsets or windows~\cite{windows}, and repeated observations of composite symbols or motif combinations~\cite{preuss,composite,motifs}.

Here, we focus on recovering a single information symbol under noiseless uniform sampling. For a linear encoding, the requested symbol $x_i$ is recoverable when the standard basis vector $\mathbf{e}_i$ belongs to the span of the sampled generator-matrix columns. Uncoded storage requires $k$ samples in expectation, whereas suitable encodings achieve smaller expectations~\cite{cover,combinatorial}. Work on this problem includes combinatorial expressions for recovery expectations~\cite{combinatorial}, geometric constructions~\cite{geometry}, asymptotic and distributional analyses~\cite{variables}, and algorithms for evaluating expectations together with improved constructions and bounds~\cite{wang}. Recent work also establishes identities relating coordinate recovery times in a code and its dual~\cite{matroid}.

Particularly relevant are the constructions of Boruchovsky et al.~\cite{first}, which combine repeated systematic columns with columns supported on pairs of information coordinates. Their analysis uses a graph whose edges represent these pair observations. Besides solving the dimension-two problem, they obtained an asymptotic normalized upper bound of $\pi^2/12\approx0.8225$ using constructions over sufficiently large finite fields.

Hofmeister, Bitar, and Yaakobi~\cite{hby} studied binary fully symmetric generator matrices, in which all columns of a given weight occur equally often. They connected sampling from these matrices to LT encoding and proved that, for bounded column weights, the asymptotic expected recovery time divided by $k$ is at least $\pi/4$ under peeling decoding. Their optimized degree distributions achieve approximately $0.7869$, while restricting the weights to one and two gives approximately $0.7939$. An optimal decoder can exploit linear dependencies that peeling does not resolve, so this lower bound does not immediately apply under optimal decoding.

In this work, we prove that the lower bound of $\pi/4$ holds under optimal decoding for ensembles with uniformly random supports and independently sampled nonzero coefficients over finite or infinite fields. The degree distribution may vary with $k$, provided that the column weights have a common finite bound. The proof relates coordinate recovery to rank increments and combines the sparse random matrix rank formula of Coja-Oghlan et al.~\cite{rank} with an integral inequality. For singleton--pair observations, we determine the exact limiting expectation and show that optimal and peeling decoding have the same limit. We also give explicit constructions over every fixed finite field attaining the optimized limit of approximately $0.793933445$, and positive-rate block constructions that approximate it arbitrarily closely.

\section{Notation and Problem Formulation}\label{sec:model}
For a positive integer $k$, let $[k]=\{1,\ldots,k\}$. Throughout, $\F$ is a field equipped with a $\sigma$-algebra that makes it a standard Borel space and with respect to which addition and multiplication are measurable, $\F^{*}=\F\setminus\{0\}$, $\bfe_i\in\F^k$ denotes the $i$-th standard basis vector, and all logarithms are natural. The support of a vector $\bfv\in\F^k$ is $\mathrm{supp}(\bfv)=\{i\in[k]:v_i\neq0\}$, and its weight is $|\mathrm{supp}(\bfv)|$. For vectors $\bfv_1,\ldots,\bfv_n$, we write $\langle\bfv_1,\ldots,\bfv_n\rangle$ for their span, and for a matrix $M$, $\langle M\rangle$ denotes its column span.

The observations are generated by independent random columns $\bfv_1,\bfv_2,\ldots\in\F^k$. Their distribution is determined by a probability distribution $P_k$ on $[k]$, called the \emph{degree distribution}, and a probability distribution $\mu$ on $\F^{*}$, called the \emph{coefficient distribution}. To generate $\bfv_j$, a weight $w$ is drawn according to $P_k$, a support $S\subseteq[k]$ is chosen uniformly at random among all subsets of size $w$, and the entries of $\bfv_j$ indexed by $S$ are drawn independently according to $\mu$. All other entries are zero. For $|\F|>2$, the support alone does not determine the coefficients, as distinct columns of a fixed generator matrix may share a support.

For $\F=\mathbb{F}_2$, every nonzero entry equals $1$, so $\bfv_j$ is determined by its support. This case describes sampling from a fixed binary generator matrix. Following~\cite{hby}, a  binary $k\times n$ matrix $G$ is called \emph{fully symmetric} if, for every $w\in[k]$, all vectors of weight $w$ appear as columns of $G$ equally often. As observed in~\cite{hby}, a uniformly sampled column of such a matrix has weight $w$ with probability equal to the fraction of columns of $G$ of weight $w$, and, given its weight, it is uniform among all vectors of that weight. Hence, sampling the columns of $G$ uniformly with replacement is the case $\F=\mathbb{F}_2$ of our model, with $P_k$ equal to the column-weight distribution of $G$.

Let $\mathbf{x}\in\F^k$ denote the information vector. The $j$-th observation is $\mathbf{x}^{\mathsf T}\bfv_j$, so after $n$ observations the decoder knows $\mathbf{x}^{\mathsf T}M_k(n)$, where $M_k(n)=[\bfv_1\ \cdots\ \bfv_n]$ is the $k\times n$ matrix formed by the first $n$ columns, and $\spn{M_k(0)}=\{\mathbf{0}\}$. The symbol $x_i$ is determined by $\mathbf{x}^{\mathsf T}M_k(n)$ for every $\mathbf{x}\in\F^k$ if and only if $\bfe_i\in\spn{M_k(n)}$. Indeed, if $\bfe_i=M_k(n)\cdot \mathbf{a}$, then $x_i=\mathbf{x}^{\mathsf T}\cdot M_k(n)\cdot \mathbf{a}$. Otherwise, there exists $\mathbf{z}\in\F^k$ with $z_i\neq0$ and
$\mathbf{z}^{\mathsf T}\cdot M_k(n)=\mathbf{0}$. Then $\mathbf{x}$ and $\mathbf{x}+\mathbf{z}$ give the same observations, while $(\mathbf{x}+\mathbf{z})_i=x_i+z_i\neq x_i$. Accordingly, the \emph{recovery time} of the $i$-th information symbol is
\begin{equation}\label{eq:tau}
 \tau_i=\min\{n\geq 0:\bfe_i\in\spn{M_k(n)}\},
\end{equation}
with $\min\emptyset=\infty$. This is the number of observations required by an optimal decoder, without any restriction on the decoding procedure.

Since the distribution of each column is invariant under permutations of the coordinates, the distribution of $\tau_i$ does not depend on $i$. We study the \emph{normalized expected recovery time} $\E[\tau_i]/k$. For $\alpha\geq0$, let
\[
 U_k(\alpha)=\Prob\bigl\{\bfe_i\notin\spn{M_k(\lfloor\alpha k\rfloor)}\bigr\}
\]
be the probability that the $i$-th symbol is not recovered after $\lfloor\alpha k\rfloor$ observations. Its dependence on $P_k$ and $\mu$ is suppressed from the notation. Since $\tau_i$ is integer valued,
$\Prob\{\tau_i>\alpha k\}=\Prob\{\tau_i>\lfloor\alpha k\rfloor\}=U_k(\alpha)$,
and therefore
\begin{equation}\label{eq:area}
 \frac{\E[\tau_i]}{k}=\int_0^\infty\Prob\{\tau_i>\alpha k\}\,d\alpha
 =\int_0^\infty U_k(\alpha)\,d\alpha.
\end{equation}

\begin{theorem}\label{thm:main}
Let $\F$ be a field, let $\mu$ be a probability distribution on $\F^{*}$, and let $W$ be a positive integer. For every sequence $(P_k)$ of degree distributions such that every column has weight at most $W$, that is, $P_k(w)=0$ for all $w>W$,
\begin{equation}\label{eq:main}
 \liminf_{k\to\infty}\frac{\E[\tau_i]}{k}\geq\frac{\pi}{4}.
\end{equation}
\end{theorem}

For $\F=\mathbb{F}_2$, Theorem~\ref{thm:main} applies to every sequence of binary fully symmetric generator matrices whose columns have weight at most $W$. Since an optimal decoder recovers $x_i$ no later than any other decoder, the theorem strengthens the bound of~\cite{hby}, which was proved for the peeling decoder.

The proof has two parts. In Section~\ref{sec:rank}, we bound $U_k(\alpha)$ from below, as $k\to\infty$, by an explicit function obtained from the rank of a sparse random matrix. In Section~\ref{sec:integral}, we show that the integral of this function is at least $\pi/4$. Both parts are first carried out for a fixed degree distribution $P=(p_1,\ldots,p_W)$ on $[W]$, which we identify with its probability generating polynomial $P(z)$.
The extension to sequences $(P_k)$ is given at the end of Section~\ref{sec:integral}.

\section{A Lower Bound on the Probability of Non-Recovery}\label{sec:rank}
Throughout this section and Section~\ref{sec:integral}, except in the proof of Theorem~\ref{thm:main}, the degree distribution is fixed. That is, $P_k=P$ for all $k\geq W$, where $P=(p_1,\ldots,p_W)$ is a probability distribution on $[W]$. We identify $P$ with its probability generating polynomial
\[
 P(z)=\sum_{w=1}^{W}p_wz^w,\qquad z\in[0,1],
\]
and denote its derivative by $P'(z)=\sum_{w=1}^{W}wp_wz^{w-1}$. Note that $P(0)=0$ and $P(1)=1$.

Our lower bound on $U_k(\alpha)$ compares the rank of $M_k(\lfloor\alpha k\rfloor)$ with the rank of an augmented matrix. Fix $\alpha,\beta\geq0$, and let $N_1,\ldots,N_k$ be Poisson random variables with mean $\beta$, independent of each other and of $M_k(\lfloor\alpha k\rfloor)$. Let $S_k(\beta)$ be the $k\times(N_1+\cdots+N_k)$ matrix that contains $N_j$ copies of $\bfe_j$ for every $j\in[k]$. We refer to its columns as \emph{singleton columns}, and to
$\widetilde M_k(\alpha,\beta)=[\,M_k(\lfloor\alpha k\rfloor)\ \ S_k(\beta)\,]$ as the \emph{augmented matrix}. Define
\[
R_k(\alpha,\beta)=\frac1k\,\E\bigl[\rank\widetilde M_k(\alpha,\beta)\bigr].
\]
In particular, $R_k(\alpha,0)=\frac1k\E[\rank M_k(\lfloor\alpha k\rfloor)]$. Repeated copies of $\bfe_j$ do not change the column span, so only the event $\{N_j\geq1\}$, which has probability $1-e^{-\beta}$, matters for the rank. 

\begin{lemma}\label{lem:finite}
For every $\alpha\geq0$ and $\beta>0$,
\begin{equation}\label{eq:finite}
 U_k(\alpha)\geq
 \frac{R_k(\alpha,\beta)-R_k(\alpha,0)}{1-e^{-\beta}}.
\end{equation}
\end{lemma}
\begin{proof}
Write $M=M_k(\lfloor\alpha k\rfloor)$, and let $J=\{j\in[k]:N_j\geq1\}$ be the set of coordinates that receive at least one singleton column. Appending a column $\bfe_j$ increases the rank by at most one, and by zero if $\bfe_j\in\spn{M}$. Hence, appending the singleton columns increases the rank of $M$ by at most $|\{j\in J:\bfe_j\notin\spn{M}\}|$. Since $J$ is independent of $M$, and each $j\in[k]$ belongs to $J$ with probability $1-e^{-\beta}$, taking expectations gives
\begin{align*}
 k\bigl(R_k(\alpha,\beta)-R_k(\alpha,0)\bigr)
 &=\E\bigl[\rank\widetilde M_k(\alpha,\beta)-\rank M\bigr]\\
 &\leq\E\Bigl[\sum_{j=1}^{k}\mathbf{1}_{\{j\in J\}}\cdot \,
       \mathbf{1}_{\{\bfe_j\notin\spn{M}\}}\Bigr]\\
 &=\sum_{j=1}^{k}\Prob\{j\in J\}\cdot \,\Prob\{\bfe_j\notin\spn{M}\}\\
 &=(1-e^{-\beta})\,kU_k(\alpha).
\end{align*}
\end{proof}

For $z\in[0,1]$ and $\alpha,\beta\geq0$, let
\begin{align}
 &H(z)=1-P(z)-(1-z)P'(z),\label{eq:H}\\
 &\Phi_{\alpha,\beta}(z)=e^{-\beta-\alpha P'(z)}-\alpha H(z),\label{eq:Phi}\\
 &R_P(\alpha,\beta)=1-\max_{0\leq z\leq1}\Phi_{\alpha,\beta}(z).\label{eq:rankvalue}
\end{align}
We show that $R_k(\alpha,\beta)\to R_P(\alpha,\beta)$ using a rank formula of Coja-Oghlan et al.~\cite{rank}, which holds over every field and for every distribution of the nonzero entries. This formula applies to random sparse matrices whose support is generated differently from our model. In our model, each column chooses its support independently, and the number of nonzero entries in each row is determined by these choices. In~\cite{rank}, one first fixes the number of nonzero entries in every row and every column, and then chooses the support uniformly at random among all patterns with these row and column counts.

Specifically, for $\alpha\geq0$ and $\beta>0$, let $B_k$ be the random $k\times N$ matrix over $\F$ generated as follows. The number of columns is $N\sim\Pois((\alpha+\beta)k)$. For each column $j\in[N]$, the number $w_j$ of its nonzero entries is drawn according to
 \begin{equation}\label{eq:mixweights}
    \Prob\{w_j=1\}=\frac{\beta+\alpha p_1}{\alpha+\beta}\quad\text{and}\quad
    \Prob\{w_j=w\}=\frac{\alpha p_w}{\alpha+\beta},
   \end{equation}
for any $2\leq w\leq W$.
 For each row $i\in[k]$, the number of its nonzero entries $d_i$ is drawn according to $\Pois(\alpha P'(1)+\beta)$. All these numbers are drawn independently. If $d_1+\cdots+d_k\neq w_1+\cdots+w_N$, or if no $k\times N$ zero-one pattern has $d_i$ ones in row $i$ and $w_j$ ones in column $j$ for all $i,j$, the numbers are drawn again. Otherwise, the support of $B_k$ is chosen uniformly at random among all such patterns, and the entries on the support are drawn independently according to $\mu$.

\begin{claim}\label{clm:counts}
For fixed $\alpha\geq0$ and $\beta>0$, as $k\to\infty$, the column-weight proportions of $\widetilde M_k(\alpha,\beta)$ converge in probability to the probabilities in~\eqref{eq:mixweights}. The number of nonzero entries in any fixed row converges in distribution to $\Pois(\alpha P'(1)+\beta)$.
\end{claim}
\begin{proof}
The number of original columns of weight $w$, divided by $k$, converges in probability to $\alpha p_w$. The number of appended singleton columns, divided by $k$, converges in probability to $\beta$. Since the total number of columns divided by $k$ converges to $\alpha+\beta>0$, their proportions give~\eqref{eq:mixweights}.

Each original column independently has a nonzero entry in a fixed row with probability $\sum_w p_w w/k=P'(1)/k$. Thus its row count is binomial with parameters $\lfloor\alpha k\rfloor$ and $P'(1)/k$, and converges to $\Pois(\alpha P'(1))$. Adding the independent $\Pois(\beta)$ singleton count gives the asserted limit.
\end{proof}

The dependence between row counts and the conditioning in $B_k$ are handled in the proof of Lemma~\ref{lem:ranklimit}.

\begin{theorem}[{\cite[Theorem~1.1]{rank}}]\label{thm:coja}
For every field $\F$ defined as in Section~\ref{sec:model}, every distribution $\mu$ on $\F^{*}$, and all $\alpha\geq0$ and $\beta>0$,
\[
 \frac{\rank B_k}{k}\longrightarrow R_P(\alpha,\beta)
 \quad\text{in probability as }k\to\infty.
\]
\end{theorem}

Theorem~\ref{thm:coja} is~\cite[Theorem~1.1]{rank} applied to the transpose of $B_k$, after simplifying the limit. The assumption $\beta>0$ ensures the divisibility condition required there, and independent nonzero entries are
covered by~\cite[Section~1.2, footnote~2]{rank}.

\begin{lemma}\label{lem:ranklimit}
For every fixed $\alpha,\beta\geq0$,
\begin{equation}\label{eq:ranklimit}
 \lim_{k\to\infty}R_k(\alpha,\beta)=R_P(\alpha,\beta).
\end{equation}
\end{lemma}

\begin{proof}
Assume first that $\beta>0$. The augmented matrix differs from $B_k$ in two ways. First, the number of columns taken from $M_k$ is fixed rather than Poisson. Second, its support is generated column by column whereas the support of $B_k$ is chosen uniformly given all row and column counts.  We
remove the first difference by passing to a Poisson number of columns, and then show that the second difference does not affect the limit. The case $\beta=0$ is treated at the end.

First, let $\widetilde M_k'(\alpha,\beta)$ be defined as $\widetilde M_k(\alpha,\beta)$, except that the number $X$ of columns taken
from $M_k$ is an independent $\Pois(\alpha k)$ random variable instead of
$\lfloor\alpha k\rfloor$. Since adding or removing a column changes the rank
by at most one, the expected ranks of $\widetilde M_k(\alpha,\beta)$ and
$\widetilde M_k'(\alpha,\beta)$ differ by at most
\begin{align}\label{eq:jensen}
 \E\bigl|X-\lfloor\alpha k\rfloor\bigr|\leq\E|X-\alpha k|+1\leq\sqrt{\alpha k}+1,
\end{align}
where the last step follows from Cauchy--Schwarz and $\operatorname{Var}(X)=\alpha k$.

Now let $A_k$ be the random matrix of Section~\ref{sec:model} with a
$\Pois((\alpha+\beta)k)$ number of columns, coefficient distribution $\mu$, and column weights distributed according to~\eqref{eq:mixweights}.
A column of $A_k$ can be generated by choosing, with probability $\alpha/(\alpha+\beta)$,
a column whose weight is distributed according to $P$, and otherwise a weight-one column. Indeed, this gives weight one with probability
$(\beta+\alpha p_1)/(\alpha+\beta)$ and weight $w\geq2$ with probability $\alpha p_w/(\alpha+\beta)$. Note that $A_k$ can equivalently be obtained from $\widetilde M_k'(\alpha,\beta)$ by independently rescaling the added singleton columns by samples from $\mu$ and uniformly permuting all columns, so the two matrices have the same rank distribution.

By the splitting property of the Poisson distribution if a $\Pois(\lambda)$ number of items are independently assigned types, where each item has type $t$ with probability $\pi_t$, then the numbers of items of the different types are independent, and the number of items of type $t$ is $\Pois(\lambda\pi_t)$. In $A_k$, the number of columns is $\Pois((\alpha+\beta)k)$, and each column independently is of the first type with probability $\alpha/(\alpha+\beta)$, or is a weight-one column with its nonzero entry at coordinate $j$ with probability $\beta/((\alpha+\beta)k)$, for each $j\in[k]$. Applying the splitting property with these $k+1$ types, the number of columns of the first type is $\Pois(\alpha k)$, the numbers of weight-one columns at the coordinates $j\in[k]$ are $\Pois(\beta)$, and all
these numbers are independent.
Hence, $A_k$ and $\widetilde M_k'(\alpha,\beta)$
have the same distribution of columns, up to the order of the columns and the values of the nonzero entries of the weight-one columns of the second type, neither of
which affects the rank.
Combining this with~\eqref{eq:jensen} gives
\[
 \Bigl|R_k(\alpha,\beta)-\frac1k\E[\rank A_k]\Bigr|\leq\frac{\sqrt{\alpha k}+1}{k}
 \longrightarrow0.
\]
Therefore, it suffices to show that $\frac1k\E[\rank A_k]\to R_P(\alpha,\beta)$.
Since $0\leq\rank A_k/k\leq1$, this follows, by the bounded convergence
theorem, once we show that $\rank A_k/k\to R_P(\alpha,\beta)$ in probability.

Second, $A_k$ differs from $B_k$ only in how the support is generated. More precisely, in $A_k$ each column chooses its support independently, whereas in $B_k$ the support is chosen uniformly given all row and column counts. We show that this does not affect the limit. We refer to number of columns, the column weights, and the row weights of a matrix together as the \emph{degree data}.
Given the degree data, $A_k$ and $B_k$ have the same distribution. For $B_k$, this holds by definition. For $A_k$, given the column weights, every zero-one pattern with these column weights is the support with the same probability $\prod_{j=1}^N\binom{k}{w_j}^{-1}$, so given also the row weights, the
support is uniform among all patterns with these weights. In both matrices,
the entries on the support are independent samples from $\mu$. Hence, the
two matrices differ only in the distribution of their degree data.
Corollary~\ref{cor:compare} in the Appendix shows that this difference does
not affect the convergence in Theorem~\ref{thm:coja}, that is,
$\rank A_k/k\to R_P(\alpha,\beta)$ in probability.

Finally, let $\beta=0$. The idea is to approximate $\beta=0$ by a small $\beta'>0$, for which the lemma is already proved, using that a few singleton columns change both $R_k$ and $R_P$ only slightly. Fix $\beta'>0$.

First, appending the singleton columns does not decrease the rank, and increases it by at most the number of coordinates $j$ with $N_j\geq1$, whose expectation is $k(1-e^{-\beta'})$. Hence,
\[
 0\leq R_k(\alpha,\beta')-R_k(\alpha,0)\leq1-e^{-\beta'}.
\]
Second, for every $z\in[0,1]$,
$$\Phi_{\alpha,0}(z)-\Phi_{\alpha,\beta'}(z)=(1-e^{-\beta'})e^{-\alpha P'(z)},$$
which lies in $[0,1-e^{-\beta'}]$. Hence, we obtain
\[
 0\leq R_P(\alpha,\beta')-R_P(\alpha,0)\leq1-e^{-\beta'}.
\]
Combining these two bounds with the triangle inequality,
\begin{align*}
 |R_k(\alpha,0)-R_P(\alpha,0)|
 &\leq|R_k(\alpha,0)-R_k(\alpha,\beta')|\\
      &\quad+|R_k(\alpha,\beta')-R_P(\alpha,\beta')|\\
 &\quad+|R_P(\alpha,\beta')-R_P(\alpha,0)|\\
 &\leq|R_k(\alpha,\beta')-R_P(\alpha,\beta')|+2(1-e^{-\beta'}).
\end{align*}
By the case $\beta>0$, the first term on the right tends to zero as
$k\to\infty$. Hence,
$\limsup_{k\to\infty}|R_k(\alpha,0)-R_P(\alpha,0)|\leq2(1-e^{-\beta'})$, and
letting $\beta'\downarrow0$ completes the proof.
\end{proof}

Let $\Phi_\alpha=\Phi_{\alpha,0}$, and define
\begin{align}
 \cA_P(\alpha)&=\argmax_{z\in[0,1]}\Phi_\alpha(z),\nonumber\\
 \cU_P(\alpha)&=\min_{z\in\cA_P(\alpha)}e^{-\alpha P'(z)}.\label{eq:envelope}
\end{align}

\begin{lemma}\label{lem:survival}
For every fixed $\alpha\geq0$,
\begin{equation}\label{eq:survival}
 \liminf_{k\to\infty}U_k(\alpha)\geq\cU_P(\alpha).
\end{equation}
\end{lemma}

The proof is given in Appendix~\ref{app:rank}. It applies Lemmas~\ref{lem:finite} and~\ref{lem:ranklimit} for a fixed $\beta>0$, and then lets $\beta\downarrow0$, so that maximizers of $\Phi_{\alpha,\beta}$ approach a maximizer of $\Phi_\alpha$.

\section{The Integral Inequality}\label{sec:integral}

By Lemma~\ref{lem:survival}, the limit inferior of $U_k(\alpha)$ is at least
$\cU_P(\alpha)$. In this section, we show that the integral of $\cU_P$ is at
least $\pi/4$, which, together with~\eqref{eq:area}, proves
Theorem~\ref{thm:main}.

\begin{proposition}\label{prop:integral}
For every degree distribution $P$ with column weights at most $W$,
\begin{equation}\label{eq:integral}
 \int_0^\infty\cU_P(\alpha)\,d\alpha\geq\frac{\pi}{4}.
\end{equation}
\end{proposition}

If $p_1=1$, the integral equals one, as shown in the proof below. For $p_1<1$, the proof has three steps. First, we show that $\cU_P(\alpha)$ is determined by the largest maximizer of $\Phi_\alpha$. Second, we rewrite the integral of $\cU_P$ in terms of the inverse of this maximizer as a function of $\alpha$. Third, we prove an identity for this inverse, from which the bound follows by the Cauchy--Schwarz inequality. The main difficulty is that the largest maximizer may jump as $\alpha$ increases, and the identity must remain valid across such jumps. Throughout the following lemmas, we assume that $p_1<1$.

For $\alpha>0$, let
\[
\zeta(\alpha)=\max\cA_P(\alpha)
\]
be the largest maximizer of $\Phi_\alpha$ on $[0,1]$, which exists since $\cA_P(\alpha)$ is compact. The first lemma characterizes the maximizers of $\Phi_\alpha$. In particular, it shows that $\cU_P(\alpha)$ is determined by $\zeta(\alpha)$.

\begin{lemma}\label{lem:maximizers}
Let $p_1<1$ and $\alpha>0$. Every maximizer $z$ of $\Phi_\alpha$ satisfies $z<1$ and $\alpha P'(z)=-\log(1-z)$. Moreover,
\begin{equation}\label{eq:zeta}
 \cU_P(\alpha)=1-\zeta(\alpha).
\end{equation}
\end{lemma}

\begin{proof}[Proof sketch]
Assume for this sketch that the maximizers lie in $(0,1)$; the endpoint cases are checked in Appendix~\ref{app:integral}. Since $H'(z)=-(1-z)P''(z)$, differentiation gives
\[
 \Phi_\alpha'(z)=\alpha P''(z)\bigl(1-z-e^{-\alpha P'(z)}\bigr).
\]
Because $p_1<1$, we have $P''(z)>0$ for $0<z<1$. Hence, at any interior maximizer $z^*$, the condition $\Phi_\alpha'(z^*)=0$ implies
\[
 1-z^*=e^{-\alpha P'(z^*)},
 \qquad \alpha P'(z^*)=-\log(1-z^*).
\]
Substituting this identity into the definition of $\cU_P(\alpha)$ gives
\begin{align*}
 \cU_P(\alpha)
 &=\min_{z\in\cA_P(\alpha)}e^{-\alpha P'(z)}\\
 &=\min_{z\in\cA_P(\alpha)}(1-z)
 =1-\zeta(\alpha),
\end{align*}
since $\zeta(\alpha)$ is the largest maximizer.
\end{proof}

By~\eqref{eq:zeta}, the integral in~\eqref{eq:integral} depends only on the function $\zeta$. The next lemma collects the properties of $\zeta$ that we need. The proofs of both lemmas are given in the Appendix.

\begin{lemma}\label{lem:zeta}
Let $p_1<1$. The function $\zeta$ is nondecreasing and right-continuous on $(0,\infty)$, and $\zeta(\alpha)\to0$ as $\alpha\downarrow0$ and $\zeta(\alpha)\to1$ as $\alpha\to\infty$.
\end{lemma}

By Lemma~\ref{lem:zeta}, $\zeta$ increases from $0$ to $1$, possibly with
jumps. This allows us to describe $\zeta$ through its inverse. For $0<z<1$, let
\[
 \sigma(z)=\inf\{\alpha>0:\zeta(\alpha)\geq z\}
\]
be the generalized inverse of $\zeta$, that is, the first value of $\alpha$ at which $\zeta$ reaches the level $z$.
By Lemma~\ref{lem:zeta}, $\sigma(z)$ is finite and positive for every $z\in(0,1)$. In the proof of Proposition~\ref{prop:integral}, we show that the integral of $\cU_P$ equals the integral of $\sigma$. The key step is the following identity. At levels where $\zeta$ is continuous, it follows from the condition $\alpha P'(z)=-\log(1-z)$ in Lemma~\ref{lem:maximizers}. Across a jump of $\zeta$, it follows from the fact that both endpoints of the jump maximize $\Phi_\alpha$ at the time of the jump.

\begin{lemma}\label{lem:normalization}
Let $p_1<1$. Then
\begin{equation}\label{eq:normalization}
 \int_0^1\frac{-\log(1-z)}{\sigma(z)}\,dz=1.
\end{equation}
\end{lemma}

The proof of Lemma~\ref{lem:normalization} is given in Appendix~\ref{app:integral}. We can now prove Proposition~\ref{prop:integral}.

\begin{proof}[Proof of Proposition~\ref{prop:integral}]
If $p_1=1$, i.e., $P(z)=z$, then $\cU_P(\alpha)=e^{-\alpha}$, and 
\[
\int_0^\infty\cU_P(\alpha)\,d\alpha = \int_0^\infty e^{-\alpha}d\alpha = 1 > \frac{\pi}{4}.
\]
Now assume $p_1<1$. We first express the integral of $\cU_P$ through $\sigma$. For every $x\in[0,1]$, the indicator $\mathbf{1}\{x<z\}$ equals one exactly for $z\in(x,1]$, so
\[
 \int_0^1\mathbf{1}\{x<z\}\,dz=1-x.
\]
We apply this with $x=\zeta(\alpha)$ and use~\eqref{eq:zeta}. Since $\zeta$ is nondecreasing and right-continuous by Lemma~\ref{lem:zeta}, we have $\{\alpha>0:\zeta(\alpha)<z\}=(0,\sigma(z))$ for every $z\in(0,1)$. Hence, by Tonelli's theorem,
\begin{equation}\label{eq:firstpassage}
\int_0^\infty\cU_P(\alpha)\,d\alpha
 =\int_0^1\int_0^\infty\mathbf{1}\{\zeta(\alpha)<z\}\,d\alpha\,dz
 =\int_0^1\sigma(z)\,dz.
\end{equation}

If $\int_0^1\sigma(z)\,dz=\infty$, then~\eqref{eq:integral} holds by~\eqref{eq:firstpassage}. Otherwise, we write
\[
 \sqrt{-\log(1-z)}=\sqrt{\sigma(z)}\cdot\sqrt{-\log(1-z)/\sigma(z)}
\]
and apply the Cauchy--Schwarz inequality together with
Lemma~\ref{lem:normalization}:
\begin{align*}
\Bigl(\int_0^1\hspace{-1.5ex}\sqrt{-\log(1\hspace{-0.5ex}-\hspace{-0.5ex}z)}\,dz\Bigr)^2
&\hspace{-1ex}\leq\hspace{-0.5ex}\Bigl(\int_0^1\hspace{-1ex}\sigma(z)\,dz\Bigr)
       \Bigl(\int_0^1\frac{-\log(1\hspace{-0.5ex}-\hspace{-0.5ex}z)}{\sigma(z)}\,dz\Bigr)\\
 &=\int_0^1\sigma(z)\,dz.
\end{align*}
Finally, the substitution $t=-\log(1-z)$, under which ${dz=e^{-t}\,dt}$ and
$z\in(0,1)$ corresponds to $t\in(0,\infty)$, gives
\[
 \int_0^1\sqrt{-\log(1-z)}\,dz=\int_0^\infty\sqrt{t}\,e^{-t}\,dt
 =\Gamma\bigl(\tfrac32\bigr)=\frac{\sqrt\pi}{2},
\]
where $\Gamma(s)=\int_0^\infty t^{s-1}e^{-t}\,dt$ is the Gamma function.
Hence, by~\eqref{eq:firstpassage},
$\int_0^\infty\cU_P(\alpha)\,d\alpha=\int_0^1\sigma(z)\,dz\geq\pi/4$, which
proves~\eqref{eq:integral}.
\end{proof}

\begin{proof}[Proof of Theorem~\ref{thm:main}]
We first prove the theorem when the degree distribution does not depend on $k$, and then reduce the general case to this one.

Let $P_k=P$ for all $k$. By~\eqref{eq:area},
\[
\frac{\E[\tau_i]}{k}=\int_0^\infty U_k(\alpha)\,d\alpha.
\]
Since $U_k\geq0$, Fatou's
lemma, Lemma~\ref{lem:survival}, and Proposition~\ref{prop:integral} give
\[
\liminf_{k\to\infty}\frac{\E[\tau_i]}{k}
\geq\int_0^\infty\liminf_{k\to\infty}U_k(\alpha)\,d\alpha \geq\int_0^\infty\cU_P(\alpha)\,d\alpha\geq\frac{\pi}{4}.
\]

Now let $(P_k)$ be any sequence of degree distributions with column weights at most $W$, and let $L=\liminf_{k\to\infty}\E[\tau_i]/k$. As $L=\infty$ is trivial, we assume $L<\infty$. Choose indices $k_1<k_2<\cdots$ with $\E[\tau_i]/k_j\to L$. Each $P_k$ is a vector in $[0,1]^W$, so by the Bolzano--Weierstrass theorem, a further subsequence satisfies $P_k\to P$ coordinatewise, where $P$ is again a probability distribution on $[W]$. Along this further subsequence, $\E[\tau_i]/k$ still converges to $L$. Hence, it suffices to show that
$\liminf U_k(\alpha)\geq\cU_P(\alpha)$ along it for every $\alpha\geq0$, since Fatou's lemma and Proposition~\ref{prop:integral} then give $L\geq\pi/4$ as in the case of a fixed $P$.

For a degree distribution $Q$, write $R_k^{(Q)}(\alpha,\beta)$ for $R_k(\alpha,\beta)$ computed with degree distribution $Q$. We couple the augmented matrices for $P_k$ and $P$ as follows. The weights of each pair of corresponding columns are drawn so that they differ with probability $d_{\mathrm{TV}}(P_k,P)$, which is possible by the definition of the total variation distance. Whenever the weights agree, the two columns have the same support and the same entries, and the two matrices have the same singleton columns. Then the two augmented matrices differ in at most $\lfloor\alpha k\rfloor$ columns, each with probability $d_{\mathrm{TV}}(P_k,P)$. Since replacing one column changes the rank by at most one,
\begin{align*}
 \bigl|R_k^{(P_k)}(\alpha,\beta)-R_k^{(P)}(\alpha,\beta)\bigr|
 &\leq\frac{\lfloor\alpha k\rfloor}{k}\,d_{\mathrm{TV}}(P_k,P)\\ &\leq\alpha\,d_{\mathrm{TV}}(P_k,P)\longrightarrow0.
\end{align*}
Together with Lemma~\ref{lem:ranklimit}, applied to the fixed distribution $P$, this gives $R_k^{(P_k)}(\alpha,\beta)\to R_P(\alpha,\beta)$ along the
subsequence, for all $\alpha,\beta\geq0$. Lemma~\ref{lem:finite} holds for
every $k$ and every degree distribution, and the proof of
Lemma~\ref{lem:survival} uses only Lemma~\ref{lem:finite} and this
convergence. Hence, $\liminf U_k(\alpha)\geq\cU_P(\alpha)$ along the
subsequence, which completes the proof.
\end{proof}

The last Cauchy--Schwarz step in the proof of Proposition~\ref{prop:integral} is the same integral estimate used
in~\cite[Lemma~5]{hby}. Here, the additional identity~\eqref{eq:jump} ensures
that~\eqref{eq:normalization} remains valid across levels skipped by jumps.

\section{A Singleton--Pair Code Construction}\label{sec:pairs}
Theorem~\ref{thm:main} shows that no ensemble with bounded column weights has normalized expected recovery time asymptotically below $\pi/4$. In this section, we show that this bound is nearly attained by explicit generator matrices. We consider the degree distribution with $p_1=s$ and $p_2=1-s$ for some $0<s<1$, that is, $P(z)=sz+(1-s)z^2$, so that every column has one or two nonzero entries. For the peeling decoder and binary matrices, this distribution was analyzed in~\cite{hby}, where the best choice of $s$ gives a normalized expected recovery time of approximately $0.7939$. Here, we extend this result in two directions. First, we show that an optimal decoder achieves the same value, over every field and for every coefficient distribution, so that optimal decoding gives no first-order improvement for this distribution. Second, we construct explicit generator matrices over every fixed finite field that achieve this value, including a block construction of positive rate.

\subsection{Asymptotic analysis}
Throughout this subsection, $P(z)=sz+(1-s)z^2$ for a fixed $0<s<1$, and the field $\F$, as in Section~\ref{sec:model}, and the coefficient distribution $\mu$ are arbitrary. To analyze the recovery time, we represent the supports of the columns by a graph $\mathcal{G}$ with vertex set $[k]$. Every column with two nonzero entries, in positions $u$ and $v$, is an edge between $u$ and $v$, and columns with the same support give parallel edges. Every column with one nonzero entry, in position $u$, \emph{marks} the vertex $u$. Marks do not affect connectivity, and a vertex without edges forms a component by
itself.

Columns supported in different components of $\mathcal{G}$ have disjoint supports. Hence, the column span is the direct sum of the spans of the columns supported in the individual components, and whether $x_i$ can be recovered depends only on the columns supported in the component of $i$, including the values of their nonzero entries. The next lemma shows that in two cases this is determined by $\mathcal{G}$ alone. In the remaining case of an unmarked component containing a cycle, recovery depends on the coefficients. For example, the columns $(1,2,0,\ldots,0)^{\mathsf T}$ and $(1,1,0,\ldots,0)^{\mathsf T}$ form a double edge and span $\bfe_1$ and $\bfe_2$, whereas two copies of $(1,1,0,\ldots,0)^{\mathsf T}$ span neither. We show below that this case has vanishing probability.

\begin{lemma}\label{lem:components}
Let $\mathcal{H}$ be a component of $\mathcal{G}$.
\begin{enumerate}
\item [(i)]If $\mathcal{H}$ contains a marked vertex, then $x_i$ can be recovered for every $i\in\mathcal{H}$.
\item[(ii)] If $\mathcal{H}$ is an unmarked tree, then $x_i$ cannot be recovered for any $i\in\mathcal{H}$.
\end{enumerate}
\end{lemma}

\begin{proof}
For (i), let $u\in\mathcal{H}$ be marked. A weight-one column at $u$ is a nonzero multiple of $\bfe_u$, so $\bfe_u$ lies in the column span. If $\bfe_w$ lies in the span and an edge between $w$ and $v$ corresponds to the column $a\bfe_w+b\bfe_v$ with $a,b\neq0$, then $\bfe_v=b^{-1}(a\bfe_w+b\bfe_v-a\bfe_w)$ also lies in the span. Since $\mathcal{H}$ is connected, repeating this along paths from $u$ shows that $\bfe_v$ lies in the span for every $v\in\mathcal{H}$.

For (ii), it suffices to find $\mathbf{y}\in\F^k$ with $\mathbf{y}^{\mathsf T}\bfv_j=0$ for every observed column $\bfv_j$ and $y_i\neq0$ for every $i\in\mathcal{H}$. Indeed, every vector $\mathbf{u}$ in the column span then satisfies $\mathbf{y}^{\mathsf T}\mathbf{u}=0$, whereas $\mathbf{y}^{\mathsf T}\bfe_i=y_i\neq0$, so $\bfe_i$ is not in the span. 

Choose a root $r\in\mathcal{H}$, set $y_r=1$, and set $y_v=0$ for $v\notin\mathcal{H}$. For the remaining vertices $v\in\mathcal{H}$, in order of increasing distance from $r$, let $w$ be the parent of $v$, let $a\bfe_w+b\bfe_v$ be the column of the edge between them, and set $y_v=-ay_w/b$. Since $a,b\neq0$, all $y_v$ with $v\in\mathcal{H}$ are nonzero. We check that $\mathbf{y}^{\mathsf T}\bfv_j=0$ for every column $\bfv_j$. If the support of $\bfv_j$ is disjoint from $\mathcal{H}$, this holds since $\mathbf{y}$ vanishes outside $\mathcal{H}$. Otherwise, $\bfv_j$ is supported in $\mathcal{H}$. Since $\mathcal{H}$ is unmarked, $\bfv_j$ is not a
weight-one column, so it corresponds to an edge of $\mathcal{H}$, that is,
$\bfv_j=a\bfe_w+b\bfe_v$, where $w$ is the parent of $v$. Then
\[
 \mathbf{y}^{\mathsf T}\bfv_j=ay_w+by_v=ay_w-ay_w=0.
\]
\end{proof}

\begin{proposition}\label{prop:pairs}
Let $P(z)=sz+(1-s)z^2$ with $0<s<1$. For every field $\F$ as in Section~\ref{sec:model} and every distribution $\mu$ on~$\F^{*}$,
\begin{equation}\label{eq:pairs}
\lim_{k\to\infty}\frac{\E[\tau_i]}{k}
 =\int_0^1\frac{-\log u}{s+2(1-s)(1-u)}\,du.
\end{equation}
\end{proposition}
\begin{proof}
Let $\alpha>0$, and let $z$ be a maximizer of $\Phi_\alpha$. By Lemma~\ref{lem:maximizers}, $z<1$ and $1-z=e^{-\alpha P'(z)}$. Since
\[
P'(z)=s+2(1-s)z,
\]
the number $u=1-z$ therefore satisfies
\begin{equation}\label{eq:pairfixed}
 u=\exp\bigl(-\alpha P'(1-u)\bigr) = \exp\bigl(-\alpha(s+2(1-s)(1-u))\bigr)
\end{equation}
Moreover, $z\neq0$, since $z=0$ would give $1=e^{-\alpha P'(0)}=e^{-\alpha s}<1$, and  hence, $u\in(0,1)$. 

Taking logarithms, a number $u\in(0,1)$ satisfies~\eqref{eq:pairfixed} if and only if $\log u+\alpha(s+2(1-s)(1-u))=0$. The left side of this equation is a strictly concave function of $u$, since $\log u$ is strictly concave and the remaining term is linear in $u$. It tends to $-\infty$ as $u\downarrow0$ and equals $\alpha s>0$ at $u=1$. By the intermediate value theorem, it therefore vanishes at some $u\in(0,1)$. It vanishes at only one such $u$, since a strictly concave function that is positive at $u=1$ cannot vanish at two points of $(0,1)$. Hence,~\eqref{eq:pairfixed} has exactly one solution in~$(0,1)$, which we denote by $u_s(\alpha)$. Since every maximizer $z$ of $\Phi_\alpha$ gives a solution $u=1-z$ of~\eqref{eq:pairfixed} in $(0,1)$, the function $\Phi_\alpha$ has the unique maximizer $1-u_s(\alpha)$. In particular, $\zeta(\alpha)=1-u_s(\alpha)$, and Lemma~\ref{lem:maximizers} gives $\cU_P(\alpha)=1-\zeta(\alpha)=u_s(\alpha)$.

By~\eqref{eq:area}, $\E[\tau_i]/k=\int_0^\infty U_k(\alpha)\,d\alpha$. Since
$U_k\geq0$, Fatou's lemma and Lemma~\ref{lem:survival} give
\[
 \liminf_{k\to\infty}\frac{\E[\tau_i]}{k}
 \geq\int_0^\infty\liminf_{k\to\infty}U_k(\alpha)\,d\alpha
 \geq\int_0^\infty u_s(\alpha)\,d\alpha.
\]

For the upper bound, we use a different representation of the expectation. For a fixed number of observations, the numbers of columns supported on different coordinates or pairs of coordinates are dependent, since they sum to the number of observations. To make them independent, we consider a Poisson number of observations. For $\alpha>0$, let $N_\alpha$ be a $\Pois(\alpha k)$ random variable, independent of the columns, and let $\widetilde U_k(\alpha)$ be the probability that $x_i$ cannot be recovered from the first $N_\alpha$ observations, that is,
$\widetilde U_k(\alpha)=\sum_{n\geq0}\Prob\{N_\alpha=n\}\Prob\{\tau_i>n\}$.
Integrating over $\alpha$, we obtain
\begin{align*}
 \int_0^\infty\widetilde U_k(\alpha)\,d\alpha
 &=\sum_{n\geq0}\Prob\{\tau_i>n\}\int_0^\infty\Prob\{N_\alpha=n\}\,d\alpha\\
 &=\sum_{n\geq0}\Prob\{\tau_i>n\}\int_0^\infty e^{-\alpha k}
   \frac{(\alpha k)^n}{n!}\,d\alpha\\
 &=\frac1k\sum_{n\geq0}\Prob\{\tau_i>n\}=\frac{\E[\tau_i]}{k}.
\end{align*}
In the first step, we exchanged the sum and the integral, which is allowed
since all terms are nonnegative, and used that $\Prob\{\tau_i>n\}$ does not
depend on $\alpha$. In the third step, we substituted $x=\alpha k$ and used
$\int_0^\infty x^ne^{-x}\,dx=n!$. The last step is the tail-sum formula.

Fix $\alpha>0$, and consider the graph $\mathcal{G}$ formed by the observations made by time $\alpha$. Each observed column has one of $k+\binom k2$ types: a weight-one column at coordinate $u$, which has probability $s/k$, or a weight-two column supported on the pair $\{u,v\}$, which has probability $(1-s)/\binom k2$. Since the number of observed columns is $\Pois(\alpha k)$, the splitting property of the Poisson distribution implies that the numbers of columns of the different types are independent, with distributions $\Pois(s\alpha)$ and $\Pois(2(1-s)\alpha/(k-1))$, respectively. Consequently, every vertex is marked independently with probability $1-e^{-s\alpha}$, every pair of vertices is joined by at least one edge independently with probability $1-e^{-2(1-s)\alpha/(k-1)}$, and the marks are independent of the edges.

By Lemma~\ref{lem:components}(i), if $x_i$ cannot be recovered, then the component $\mathcal{H}_i$ of $i$ in $\mathcal{G}$ is unmarked. Given the edges, $\mathcal{H}_i$ is fixed, and each of its vertices is unmarked independently with probability $e^{-s\alpha}$. Hence, for every integer $L\geq1$,
\begin{align*}
 \widetilde U_k(\alpha)
 &\leq\E\bigl[e^{-s\alpha|\mathcal{H}_i|}\bigr]\\
 &=\sum_{r=1}^{L}e^{-s\alpha r}\Prob\{|\mathcal{H}_i|=r\}
   +\sum_{r>L}e^{-s\alpha r}\Prob\{|\mathcal{H}_i|=r\}\\
 &\leq\sum_{r=1}^{L}e^{-s\alpha r}\Prob\{|\mathcal{H}_i|=r\}
   +e^{-s\alpha(L+1)}\Prob\{|\mathcal{H}_i|>L\}\\
 &\leq\sum_{r=1}^{L}e^{-s\alpha r}\Prob\{|\mathcal{H}_i|=r\}
   +e^{-s\alpha(L+1)}.
\end{align*}
If $|\mathcal{H}_i|=r$, then for some size-$r$ set of vertices $S$, such that $i\in S$, the edges within $S$ contain a spanning tree of $S$ and no edge joins $S$ to the other vertices. There are $\binom{k-1}{r-1}$ such sets and, by Cayley's formula, $r^{r-2}$ spanning trees of each. The edges of a tree and the excluded edges lie on disjoint pairs of vertices and are therefore independent, so a union
bound gives
\begin{align*}
 \Prob\{|\mathcal{H}_i|=r\}
&\hspace{-0.5ex}\leq\hspace{-0.5ex}\binom{k\hspace{-0.5ex}-\hspace{-0.5ex}1}{r\hspace{-0.5ex}-\hspace{-0.5ex}1}r^{r-2}\Bigl(1\hspace{-0.5ex}-\hspace{-0.5ex}e^{-\frac{2(1-s)\alpha}{k-1}}\Bigr)^{r-1}
   \hspace{-0.6ex}e^{-\frac{2(1-s)\alpha r(k-r)}{k-1}}\\
&\longrightarrow\frac{r^{r-1}}{r!}\bigl(2(1-s)\alpha\bigr)^{r-1}
   e^{-2(1-s)\alpha r}
\end{align*}
as $k\to\infty$. Letting $k\to\infty$ and then $L\to\infty$, we obtain
$\limsup_{k\to\infty}\widetilde U_k(\alpha)\leq v$, where
\[
 v=\sum_{r\geq1}\frac{r^{r-1}}{r!}\bigl(2(1-s)\alpha\bigr)^{r-1}
 e^{-(2-s)\alpha r}.
\]

It remains to show that $v=u_s(\alpha)$. The power series $y(w)=\sum_{r\geq1}r^{r-1}w^r/r!$ satisfies $y(w)=we^{y(w)}$ and
$0\leq y(w)<1$ for $0\leq w<1/e$, by Lagrange inversion. Since
$2(1-s)\alpha\,v=y(w)$ with $w=2(1-s)\alpha e^{-(2-s)\alpha}<1/e$, this gives
$v=\exp(-\alpha(s+2(1-s)(1-v)))$, that is, $v$
satisfies~\eqref{eq:pairfixed}. Moreover, $0<v<1$, since $y(w)<1$ and
$y(w)e^{-y(w)}=w<2(1-s)\alpha e^{-2(1-s)\alpha}$. By uniqueness,
$v=u_s(\alpha)$. Since $\widetilde U_k(\alpha)\leq e^{-s\alpha}$, Fatou's
lemma applied to $e^{-s\alpha}-\widetilde U_k(\alpha)$ gives
$\limsup_{k\to\infty}\E[\tau_i]/k\leq\int_0^\infty u_s(\alpha)\,d\alpha$,
which matches the lower bound.

Finally, solving~\eqref{eq:pairfixed} for $\alpha$ gives the strictly
decreasing inverse $\alpha(u)=-\log u/(s+2(1-s)(1-u))$ on $(0,1)$, with
$\alpha(u)\to\infty$ as $u\downarrow0$ and $\alpha(u)\to0$ as $u\uparrow1$.
Hence, $\int_0^\infty u_s(\alpha)\,d\alpha=\int_0^1\alpha(u)\,du$, which is
the right side of~\eqref{eq:pairs}.
\end{proof}

The right side of~\eqref{eq:pairs} is minimized at $s\approx0.1554737$, where numerical integration gives the value $0.793933445$. This agrees with the optimal value for this distribution under the peeling decoder reported in~\cite{hby}. In fact, the peeling decoder recovers exactly the coordinates in marked components, so the upper bound in the proof of Proposition~\ref{prop:pairs} also applies to it. Hence, for this distribution, an optimal decoder has the same asymptotic normalized expected recovery time as the peeling decoder, for every $s$, over every field, and for every coefficient distribution.

\subsection{Finite-field generator matrices}\label{sec:construction}
The degree distribution $P(z)=sz+(1-s)z^2$ of the previous subsection can be realized by sampling a fixed generator matrix. Fix a prime power $q$ and an integer $a\geq1$. Over $\mathbb F_q$, let
\begin{equation}\label{eq:generator}
 G_{k,q}(a)=\bigl[\underbrace{I_k\ \cdots\ I_k}_{a\text{ copies}}\ \big|\ D_{k,q}\bigr],
\end{equation}
where $I_k$ is the $k\times k$ identity matrix and $D_{k,q}$ contains the column $\bfe_u+\lambda\bfe_v$ once for every $1\leq u<v\leq k$ and every $\lambda\in\mathbb F_q^*$. This is a
dimension-$k$ code of length
\[
 ak+(q-1)\binom{k}{2}.
\]
Uniform sampling of its columns with replacement  produces a weight-one column with probability 
\[
s_k=\frac{ak}{ak+(q-1)\binom{k}{2}}.\]
Otherwise, it produces a column $\bfe_u+\lambda\bfe_v$, where $\{u,v\}$ is a uniformly random pair and $\lambda$ is uniform on $\mathbb F_q^*$, independently. Unlike the random columns of Section~\ref{sec:model}, the matrix is fixed before sampling, so no new coefficients are generated during retrieval. For any prescribed $0<s<1$, we choose the number of copies of $I_k$ in
$G_{k,q}(a)$ as
\begin{equation}\label{eq:multiplicity}
a=a_k:=\left\lceil\frac{s(q-1)(k-1)}{2(1-s)}\right\rceil,
\end{equation}
which gives $s_k\to s$ as $k\to\infty$. 

Hence, for $q=2$, sampling the columns of $G_{k,2}(a_k)$ is the setting analyzed in~\cite{hby} for the peeling decoder with column weights one and two. For larger $q$, the matrices $G_{k,q}(a)$ are closely related to the construction of~\cite{first}, which also combines copies of $I_k$ with weight-two columns, but uses $x$ weight-two columns per pair, with distinct coefficient ratios, over a field whose size grows with $k$. For that construction, it was conjectured in~\cite{hby} that, with the ratio of weight-one to weight-two columns optimized, the normalized expected recovery time tends to $\approx0.7939$, the minimal value obtained in~\cite{hby} for binary fully symmetric codes with column weights one and two under the peeling decoder. Proposition~\ref{prop:finitecode} below establishes this limit for $G_{k,q}(a_k)$ over every fixed finite field.

\begin{lemma}\label{lem:field}
Let $\bfv_1,\dots,\bfv_n$ be sampled from $G_{k,q}(a)$, and let $\mathcal{H}_i$ be the component of $i$ in $\mathcal{G}$. If $\mathcal{H}_i$ is marked, then $\bfe_i\in\langle M_k(n)\rangle$. If $\mathcal{H}_i$ is unmarked and has $m$ edges, counted with multiplicity, then
\begin{equation}\label{eq:field}
 \Prob\{\bfe_i\notin\langle M_k(n)\rangle\mid\text{supports}\}
 =(q-1)^{-(m-|\mathcal{H}_i|+1)}.
\end{equation}
\end{lemma}

\begin{proof}
If $\mathcal{H}_i$ is marked, then $\bfe_i\in\langle M_k(n)\rangle$ by Lemma~\ref{lem:components}(i). Otherwise, let $r=|\mathcal{H}_i|$ and fix a spanning tree of $\mathcal{H}_i$. Its $r-1$ columns are linearly independent, so they span a subspace $T$ of dimension $r-1$ of the $r$-dimensional space of vectors supported in $\mathcal{H}_i$. Since the tree is unmarked, $T$ contains no weight-one vector by Lemma~\ref{lem:components}(ii).

Each of the other $m-r+1$ columns $\bfe_u+\lambda\bfe_v$ of $\mathcal{H}_i$ either lies in $T$, or extends $T$ to the whole $r$-dimensional space, which contains $\bfe_i$. Hence, $\bfe_i\notin\langle M_k(n)\rangle$ if and only if all of them lie in $T$. For given $u,v$, exactly one $\lambda\in\mathbb F_q^*$ puts $\bfe_u+\lambda\bfe_v$ in $T$. At least one does, since combining the tree columns along the path from $u$ to $v$ gives such a vector. At most one does, since otherwise subtracting two of them would put a weight-one vector in $T$. Since the $\lambda$'s are independent and uniform on $\mathbb F_q^{*}$, all of them take this value with probability $(q-1)^{-(m-r+1)}$.
\end{proof}

\begin{proposition}\label{prop:finitecode}
For every prime power $q$ and $0<s<1$, sampling the columns of
$G_{k,q}(a_k)$, with $a_k$ as in~\eqref{eq:multiplicity}, gives the same
limit as in~\eqref{eq:pairs}, that is,
\begin{equation}\label{eq:finitecode}
 \lim_{k\to\infty}\frac{\E[\tau_i]}{k}
 =\int_0^1\frac{-\log u}{s+2(1-s)(1-u)}\,du.
\end{equation}
\end{proposition}

\begin{proof}
Multiplying a column by a nonzero scalar does not change the span, so it does not change $\tau_i$. Up to such scalars, a sampled column of $G_{k,q}(a_k)$ is a column of the model of Section~\ref{sec:model} over $\mathbb F_q$ with degree distribution $P_k(z)=s_kz+(1-s_k)z^2$ and $\mu$ uniform on $\mathbb F_q^{*}$. Indeed, in that model a weight-one column is a multiple of some $\bfe_u$, and a weight-two column on a uniformly random pair $u<v$ equals $a\bfe_u+b\bfe_v=a(\bfe_u+(b/a)\bfe_v)$, where $b/a$ is uniform on
$\mathbb F_q^{*}$.

Hence, it suffices to show that Proposition~\ref{prop:pairs} still holds when $s$ is replaced by $s_k\to s$. For the lower bound, the proof of Theorem~\ref{thm:main} shows that Lemma~\ref{lem:survival} holds whenever $P_k\to P$. For the upper bound, the proof of Proposition~\ref{prop:pairs} uses $s$ only through the probabilities that a vertex is marked and that a pair is joined by an edge, and the same limits are obtained with $s_k$ in place of
$s$.
\end{proof}

Taking $s=0.1554737$ gives the value $0.793933445\ldots$ over every fixed field. This realizes the known asymptotic constant with an explicit code; increasing $q$ improves recovery on individual cyclic components but does not change this limit. The length of $G_{k,q}(a_k)$ satisfies \[
a_kk+(q-1)\binom{k}{2}\sim\frac{(q-1)k^2}{2(1-s)},
\]
so the rate tends to zero. One may instead fix a block dimension $\ell$ and take a direct sum of $t$ copies of $G_{\ell,q}(a_\ell)$. Sampling enters the requested block with probability $1/t$, and columns of other blocks do not help to recover $x_i$.
Therefore, for $i\in[\ell]$,
\[
\frac{\E[\tau_i(\bigoplus_{j=1}^tG_{\ell,q}(a_\ell))]}{t\ell}
=\frac{\E[\tau_i(G_{\ell,q}(a_\ell))]}{\ell}.
\]
This gives a positive rate for fixed $\ell$, and a normalized expectation arbitrarily close to the right side of~\eqref{eq:pairs} by choosing $\ell$ sufficiently large. The block construction is not an instance of the model of Section~\ref{sec:model}, since its weight-two columns only join coordinates in
the same block.

\section{Conclusion}
We proved that $\liminf_{k\to\infty}\E[\tau_i]/k\geq\pi/4$ for symmetric linear ensembles with uniformly bounded column weights. This bound holds under optimal decoding.

For degree distributions with support size in ${1,2}$, we determined the limiting normalized expected recovery time. The limit is the same under optimal and peeling decoding, and its minimum over these distributions is approximately $0.793933445$. We gave explicit code families that attain this limit over every fixed finite field. We also showed that positive-rate block constructions can achieve normalized expected recovery times arbitrarily close to this value.

We have also obtained several results for columns with growing support over infinite fields, which we plan to include in an extended version.


\section*{AI Statement}
OpenAI ChatGPT and Anthropic Claude were used as writing assistants for editing
and rephrasing the text, including the presentation of some proofs, and for
checking related literature. All results and proofs are due to the authors,
who directed and verified all changes

\newpage
\appendices

\section{Proofs for Section~\ref{sec:rank}}\label{app:rank}


Let $\alpha\geq0$ and $\beta>0$, and let $A_k$ and $B_k$ be as in the proof
of Lemma~\ref{lem:ranklimit}. For a matrix with $N$ columns, column weights
$w_1,\ldots,w_N$, and row weights $d_1,\ldots,d_k$, write 
$$a=(N,w_1,\ldots,w_N,d_1,\ldots,d_k)$$ 
for its degree data and $T=w_1+\cdots+w_N$ for its number of nonzero entries.

\begin{lemma}\label{lem:compare}
For every $\delta>0$ there are constants $C_\delta$ and $k_\delta$ such that,
for all $k\geq k_\delta$ and every set $\mathcal{B}$ of matrices,
\[
 \Prob\{A_k\in\mathcal{B}\}\leq C_\delta\,\Prob\{B_k\in\mathcal{B}\}+\delta.
\]
\end{lemma}

\begin{proof}
Since $A_k$ and $B_k$ have the same distribution given their degree data, it
suffices to compare the probabilities $p_A(a)$ and $p_B(a)$ of the degree
data. Let $\lambda=\alpha P'(1)+\beta$ and
$\pi_k(T)=\Prob\{\Pois(\lambda k)=T\}$. Let $q(a)$ be the probability of $a$
when all numbers in the definition of $B_k$ are drawn once, without
redrawing, and let $Z_k$ be the probability that they are not redrawn. Then
$p_B(a)=q(a)/Z_k$ whenever $p_B(a)>0$. Since a sum of $k$ independent
$\Pois(\lambda)$ random variables is a $\Pois(\lambda k)$ random variable,
\begin{align*}
 Z_k & \leq\Prob\{d_1+\cdots+d_k=w_1+\cdots+w_N\}\\
 &=\sum_{T}\Prob\{w_1+\cdots+w_N=T\}\,\pi_k(T)\\
 &\leq\max_T\pi_k(T).
\end{align*}
To compute $p_A(a)$, suppose first that each column draws its $w_j$ positions
independently and uniformly with replacement. Then, given $N$ and the column
weights, the row counts are multinomial, and
\[
 \frac{T!}{d_1!\cdots d_k!}\,k^{-T}
 =\frac{1}{\pi_k(T)}\prod_{i=1}^{k}e^{-\lambda}\frac{\lambda^{d_i}}{d_i!}.
\]
The support of $A_k$ is obtained from such draws by conditioning on the event
that no column draws a position twice. Given $N$ and the column weights, this
event has probability at least $\prod_{j}\bigl(1-w_j^2/(2k)\bigr)\geq e^{-W^2N/k}$
for $k\geq W^2$. Therefore, for every $a$ with $p_A(a)>0$,
\begin{align*}
 p_A(a)& \leq e^{W^2N/k}\,\frac{q(a)}{\pi_k(T)}\\
& =e^{W^2N/k}\,\frac{Z_k}{\pi_k(T)}\,p_B(a)\\
& \leq e^{W^2N/k}\,\frac{\max_{T'}\pi_k(T')}{\pi_k(T)}\,p_B(a).
\end{align*}
By Stirling's formula, there is a constant $C_t'$ such that
$\max_{T'}\pi_k(T')\leq C_t'\,\pi_k(T)$ whenever $|T-\lambda k|\leq t\sqrt{k}$
and $k$ is large enough. Hence, on the set of degree data with
$N\leq2(\alpha+\beta)k$ and $|T-\lambda k|\leq t\sqrt{k}$, we have
$p_A(a)\leq C_t\,p_B(a)$ with $C_t=C_t'e^{2(\alpha+\beta)W^2}$, and
therefore
\begin{align*}
 \Prob\{A_k\in\mathcal{B}\}\leq C_t\,\Prob\{B_k\in\mathcal{B}\}
 &+\Prob\{N>2(\alpha+\beta)k\}\\ & \ \ \ +\Prob\bigl\{|T-\lambda k|>t\sqrt{k}\bigr\},
\end{align*}
where the last two probabilities refer to $A_k$. In $A_k$, the number of
columns $N$ is $\Pois((\alpha+\beta)k)$, and $T$ is the sum of $N$
independent column weights, so $T$ has mean $\lambda k$ and variance
$(\alpha+\beta)k\,\E[w^2]$. Chebyshev's inequality bounds the last two
probabilities by $1/((\alpha+\beta)k)$ and $(\alpha+\beta)\E[w^2]/t^2$,
respectively.
Choosing $t$ with $(\alpha+\beta)\E[w^2]/t^2\leq\delta/2$, and $k$ so large
that $1/((\alpha+\beta)k)\leq\delta/2$, proves the lemma with $C_\delta=C_t$.
\end{proof}

\begin{corollary}\label{cor:compare}
For all $\alpha\geq0$ and $\beta>0$, $\rank A_k/k\to R_P(\alpha,\beta)$ in
probability as $k\to\infty$.
\end{corollary}

\begin{proof}
Fix $\varepsilon>0$, and let $\mathcal{B}$ be the set of matrices $M$ with
\[
 \Bigl|\frac{\rank M}{k}-R_P(\alpha,\beta)\Bigr|>\varepsilon.
\]
By Theorem~\ref{thm:coja}, $\Prob\{B_k\in\mathcal{B}\}\to0$. Let $\delta>0$.
By Lemma~\ref{lem:compare}, for all $k\geq k_\delta$,
\[
 \Prob\{A_k\in\mathcal{B}\}\leq C_\delta\,\Prob\{B_k\in\mathcal{B}\}+\delta,
\]
and since $C_\delta$ does not depend on $k$,
$\limsup_{k\to\infty}\Prob\{A_k\in\mathcal{B}\}\leq\delta$. As $\delta>0$ is
arbitrary, $\Prob\{A_k\in\mathcal{B}\}\to0$, and as $\varepsilon>0$ is
arbitrary, $\rank A_k/k\to R_P(\alpha,\beta)$ in probability.
\end{proof}

\begin{proof}[Proof of Lemma~\ref{lem:survival}]

Let $\beta>0$. By Lemma~\ref{lem:finite}, for every $k$,
\[
U_k(\alpha)\geq\frac{R_k(\alpha,\beta)-R_k(\alpha,0)}{1-e^{-\beta}}.
\]
By Lemma~\ref{lem:ranklimit}, the right side converges to
$\bigl(R_P(\alpha,\beta)-R_P(\alpha,0)\bigr)/(1-e^{-\beta})$ as $k\to\infty$.
Hence, by~\eqref{eq:rankvalue},
\begin{align*}
 \liminf_{k\to\infty}U_k(\alpha)
 &\geq\frac{R_P(\alpha,\beta)-R_P(\alpha,0)}{1-e^{-\beta}}\\
 &=\frac{\max_{z\in[0,1]}\Phi_\alpha(z)-\max_{z\in[0,1]}\Phi_{\alpha,\beta}(z)}{1-e^{-\beta}}.
\end{align*}
By~\eqref{eq:Phi},
$\Phi_{\alpha,\beta}(z)=\Phi_\alpha(z)-(1-e^{-\beta})e^{-\alpha P'(z)}$.
Let $z_\beta$ be a maximizer of $\Phi_{\alpha,\beta}$ on $[0,1]$. Then
\begin{align*}
 \max_{z\in[0,1]}\Phi_{\alpha,\beta}(z)
 &=\Phi_\alpha(z_\beta)-(1-e^{-\beta})e^{-\alpha P'(z_\beta)}
\\&\leq\max_{z\in[0,1]}\Phi_\alpha(z)-(1-e^{-\beta})e^{-\alpha P'(z_\beta)}.
\end{align*}
Consequently, for every $\beta>0$ and every choice of
\[
z_\beta\in\argmax_{z\in[0,1]}\Phi_{\alpha,\beta}(z),
\]
we have
\[
\liminf_{k\to\infty}U_k(\alpha)
\geq e^{-\alpha P'(z_\beta)}.
\]

Here $z_\beta$ maximizes the perturbed function $\Phi_{\alpha,\beta}$, whereas the desired bound in the lemma statement is defined using the maximizers of $\Phi_\alpha=\Phi_{\alpha,0}$. It suffices to establish the same inequality for at least one point $z^*\in\cA_P(\alpha)$, since then
\begin{align*}
\liminf_{k\to\infty}U_k(\alpha)
&\geq e^{-\alpha P'(z^*)}\\
&\geq\min_{z\in\cA_P(\alpha)}e^{-\alpha P'(z)}
=\cU_P(\alpha).
\end{align*}

We obtain such a point by letting the perturbation vanish. Set $\beta_h=1/h$, and for each $h$, choose one maximizer $z_h$ of $\Phi_{\alpha,\beta_h}$. These points need not converge, but compactness of $[0,1]$ guarantees a subsequence $z_{h_i}\to z^*$. We next show that this limiting point maximizes $\Phi_\alpha$. Continuity of $z\mapsto e^{-\alpha P'(z)}$ will then allow us to pass the preceding inequality to the limit.

Since $0\leq\Phi_\alpha(z)-\Phi_{\alpha,\beta_h}(z)\leq1-e^{-\beta_h}$
for all $z\in[0,1]$, the maximal value of the two functions differ by at most
$1-e^{-\beta_h}$, and
\begin{align*}
\Phi_\alpha(z_h)\geq\Phi_{\alpha,\beta_h}(z_h)&=\max_{z\in[0,1]}\Phi_{\alpha,\beta_h}(z)\\
& \geq\max_{z\in[0,1]}\Phi_\alpha(z)-(1-e^{-\beta_h}).
\end{align*}
Letting $i\to\infty$ along $h=h_i$ and using the continuity of $\Phi_\alpha$ gives
$\Phi_\alpha(z^*)\geq\max_{z\in[0,1]}\Phi_\alpha(z)$, so $z^*\in\cA_P(\alpha)$.
Finally, since $\liminf_{k\to\infty}U_k(\alpha)\geq e^{-\alpha P'(z_h)}$ for
every $h$, and ${z\mapsto e^{-\alpha P'(z)}}$ is continuous,
\begin{align*}
\liminf_{k\to\infty}U_k(\alpha)&\geq \lim_{i\to\infty}e^{-\alpha P'(z_{h_i})}\\
& =e^{-\alpha P'(z^*)}\\ 
&\geq\min_{z\in\cA_P(\alpha)}e^{-\alpha P'(z)}\\ 
&=\cU_P(\alpha),
\end{align*}
which proves~\eqref{eq:survival}.
\end{proof}

\section{Proofs for Section~\ref{sec:integral}}\label{app:integral}
Throughout this subsection, $p_1<1$. 
since $\sum_{w=1}^W p_w=1$, there exists $w'\ge2$ such that $p_{w'}>0$ and hence 
\[
P''(z) = \sum_{w=2}^W w(w-1)p_w z^{w-2}\ge w'(w'-1)p_{w'}z^{w'-2}>0
\]
for $0<z\leq1$. 

\begin{proof}[Proof of Lemma~\ref{lem:maximizers}]
 Since
$P''(z)>0$ for $0<z\leq1$, the derivative $P'$ is strictly increasing on $[0,1]$, so $z\mapsto e^{-\alpha P'(z)}$ is strictly decreasing for $\alpha>0$. Hence, the minimum in the definition of $\cU_P(\alpha)$ is attained at the largest element of $\cA_P(\alpha)$, that is, $\cU_P(\alpha)=e^{-\alpha P'(\zeta(\alpha))}$.

Let $\Phi_\alpha'$ denote the derivative of $\Phi_\alpha$ with respect to $z$, then 
\begin{equation}\label{eq:stationarity}
 \Phi_\alpha'(z)=\alpha P''(z)\bigl(1-z-e^{-\alpha P'(z)}\bigr).
\end{equation}

 Next, we show that an interior maximizer satisfies ${\alpha P'(z)=-\log(1-z)}$. Indeed, for any $\alpha>0$,  let $z^*\in(0,1)$ be an interior maximizer of $\Phi_\alpha$.
As $\Phi_\alpha$ is differentiable, $\Phi_\alpha'(z^*)=0$. Since
$\alpha P''(z^*)>0$, \eqref{eq:stationarity} gives $1-z^*=e^{-\alpha P'(z^*)}$, and
taking logarithms, which is possible since $1-z^*>0$,
\[
 \alpha P'(z^*)=-\log(1-z^*).
\]

The endpoint $z=1$ does not maximize $\Phi_\alpha$. To see this, note that by \eqref{eq:stationarity},
$\Phi_\alpha'(1)=-\alpha P''(1)e^{-\alpha P'(1)}<0$. Since $\Phi_\alpha'$ is continuous, there is $\delta>0$ such that $\Phi_\alpha'<0$ on~${(1-\delta,1]}$. For~$z\in(1-\delta,1)$, the mean value theorem gives some
${\xi\in(z,1)}$ with
\[
 \Phi_\alpha(1)-\Phi_\alpha(z)=\Phi_\alpha'(\xi)(1-z)<0,
\]
since $\Phi_\alpha'(\xi)<0$ and $1-z>0$. Hence, $\Phi_\alpha(z)>\Phi_\alpha(1)$,
that is, $\Phi_\alpha$ is strictly decreasing on $(1-\delta,1]$.

Next, suppose that the endpoint $z=0$ maximizes $\Phi_\alpha$. We show that
then $P'(0)=p_1=0$. Otherwise, 
$$\lim_{z\downarrow 0} 1-z-e^{-\alpha P'(z)}=1-e^{-\alpha p_1}>0$$
so this factor is positive on $[0,\delta)$ for some
$\delta>0$. Since also $P''(z)>0$ for $z>0$, \eqref{eq:stationarity} gives
$\Phi_\alpha'>0$ on $(0,\delta)$. Take $z\in(0,\delta)$, by the mean value theorem there exists 
${\xi\in(0,z)}$ with
\[
 \Phi_\alpha(z)-\Phi_\alpha(0)=z\Phi_\alpha'(\xi)>0,
\]
and thus $\Phi_\alpha(z)>\Phi_\alpha(0)$ for $z\in(0,\delta)$, a contradiction. Hence,
$P'(0)=0$, and so $1-z=e^{-\alpha P'(z)}$ holds also at $z=0$.

Consequently, every maximizer $z^*$ of $\Phi_\alpha$ satisfies
$$1-z^*=e^{-\alpha P'(z^*)}.$$ Applying this to $z=\zeta(\alpha)$ and using
$\cU_P(\alpha)=e^{-\alpha P'(\zeta(\alpha))}$, we obtain
\[
 \cU_P(\alpha)=1-\zeta(\alpha).
\]
\end{proof}

\begin{proof}[Proof of Lemma~\ref{lem:zeta}]
We first show that $\zeta$ is nondecreasing. To this end, we compare the values of $\Phi_\alpha$ at two points $0\leq z_1<z_2\leq1$:
\begin{equation}\label{eq:crossing}
 \frac{\Phi_\alpha(z_2)-\Phi_\alpha(z_1)}{\alpha}
 =H(z_1)-H(z_2)-\int_{P'(z_1)}^{P'(z_2)}e^{-\alpha u}\,du.
\end{equation}
The right side is strictly increasing in $\alpha$. Hence, if ${\Phi_\alpha(z_2)\geq\Phi_\alpha(z_1)}$ for some $\alpha>0$, then $\Phi_{\alpha'}(z_2)>\Phi_{\alpha'}(z_1)$ for all $\alpha'>\alpha$.

We now use this to show that $\zeta$ is nondecreasing. Assume, to the contrary, that $\zeta(\alpha')<\zeta(\alpha)$ for some $0<\alpha<\alpha'$, and put $z_1=\zeta(\alpha')$ and $z_2=\zeta(\alpha)$. Since $z_2$ maximizes $\Phi_\alpha$, we have $\Phi_\alpha(z_2)\geq\Phi_\alpha(z_1)$, and therefore $\Phi_{\alpha'}(z_2)>\Phi_{\alpha'}(z_1)$. This contradicts the fact that $z_1$ maximizes $\Phi_{\alpha'}$. Therefore,~$\zeta$ is nondecreasing.

Next, we show that $\zeta$ is right-continuous. Let $\alpha_0>0$. Since~$\zeta$ is nondecreasing, the limit
$z^+=\lim_{\alpha\downarrow\alpha_0}\zeta(\alpha)$ exists and satisfies $z^+\geq\zeta(\alpha_0)$. It remains to show that $z^+\leq\zeta(\alpha_0)$. To this end, we show that $z^+$ maximizes $\Phi_{\alpha_0}$. 
Fix $y\in[0,1]$. For every $\alpha>\alpha_0$, since $\zeta(\alpha)$ maximizes
$\Phi_\alpha$,
\[
\Phi_\alpha(\zeta(\alpha))\geq\Phi_\alpha(y).
\]
Let $\alpha\downarrow\alpha_0$ on both sides. On the right side, the
point $y$ is fixed, and $\Phi_\alpha(y)=e^{-\alpha P'(y)}-\alpha H(y)$ is a continuous function of $\alpha$, so $\Phi_\alpha(y)\to\Phi_{\alpha_0}(y)$. On the left side, both the parameter $\alpha$ and the point $\zeta(\alpha)$ change, with $(\alpha,\zeta(\alpha))\to(\alpha_0,z^+)$. Since $P'$ and $H$ are polynomials, the function $(\alpha,z)\mapsto e^{-\alpha P'(z)}-\alpha H(z)$ is continuous in both variables jointly, and therefore
$\Phi_\alpha(\zeta(\alpha))\to\Phi_{\alpha_0}(z^+)$. Since the inequality holds for every $\alpha>\alpha_0$, it is preserved in the limit, and $\Phi_{\alpha_0}(z^+)\geq\Phi_{\alpha_0}(y)$. Since $y\in[0,1]$ is arbitrary, $z^+$ maximizes $\Phi_{\alpha_0}$. Every maximizer of $\Phi_{\alpha_0}$ is at most its largest maximizer $\zeta(\alpha_0)$, so $z^+\leq\zeta(\alpha_0)$. Together with $z^+\geq\zeta(\alpha_0)$, this gives $\lim_{\alpha\downarrow\alpha_0}\zeta(\alpha)=\zeta(\alpha_0)$.

We now show that $\zeta(\alpha)\to0$ as $\alpha\downarrow0$. For $x\geq0$, we have $1-x\leq e^{-x}\leq1-x+x^2/2$. Applying this with ${x=\alpha P'(z)\in[0,\alpha P'(1)]}$ gives
\[e^{-\alpha P'(z)}=1-\alpha P'(z)+r_\alpha(z)\]
with
\[0\leq r_\alpha(z)\leq \frac{\alpha^2(P'(z))^2}{2}\le \frac{\alpha^2(P'(1))^2}{2}\] 
for all $z\in[0,1]$. Since $P'(z)+H(z)=1-P(z)+zP'(z)$, it follows that
\[
 \Phi_\alpha(z)=1-\alpha g(z)+r_\alpha(z),
 \]
 with
 \[
 g(z)=1-P(z)+zP'(z).
\]
The function $g$ is strictly increasing on $[0,1]$, since $g'(z)=zP''(z)>0$ for $0<z\leq1$. Now fix $\varepsilon\in(0,1)$, and let $z\geq\varepsilon$. Then
\begin{align*}
 \Phi_\alpha(z)-\Phi_\alpha(0)
 &=-\alpha\bigl(g(z)-g(0)\bigr)+r_\alpha(z)-r_\alpha(0)\\
 &\leq-\alpha\bigl(g(\varepsilon)-g(0)\bigr)+\frac{\alpha^2(P'(1))^2}{2},
\end{align*}
where the inequality uses $g(z)\geq g(\varepsilon)$, together with ${r_\alpha(z)\leq\alpha^2(P'(1))^2/2}$ and $r_\alpha(0)\geq0$. Since $g(\varepsilon)>g(0)$, the right side is negative
whenever
\[
0<\alpha<\frac{2\bigl(g(\varepsilon)-g(0)\bigr)}{(P'(1))^2}.
\]
For such $\alpha$, every $z\geq\varepsilon$ satisfies $\Phi_\alpha(z)<\Phi_\alpha(0)$, so no $z\geq\varepsilon$ maximizes $\Phi_\alpha$, that is, $\zeta(\alpha)<\varepsilon$. Since $\varepsilon$ is arbitrary, $\zeta(\alpha)\to0$ as $\alpha\downarrow0$.

Finally, we show that $\zeta(\alpha)\to1$ as $\alpha\to\infty$. For every $z\in[0,1]$,
\[
 \frac{\Phi_\alpha(z)}{\alpha}=\frac{e^{-\alpha P'(z)}}{\alpha}-H(z),
 \]
where
\[0<\frac{e^{-\alpha P'(z)}}{\alpha}\leq\frac1\alpha,
\]
since $P'(z)\geq0$. Since $H'(z)=-(1-z)P''(z)<0$ for $0<z<1$, the function $H$ is strictly decreasing on $[0,1]$, and $H(1)=1-P(1)=0$. Now fix $\varepsilon\in(0,1)$, and let $z\leq1-\varepsilon$. Then $H(z)\geq H(1-\varepsilon)$, and therefore
\begin{align*}
 \frac{\Phi_\alpha(z)}{\alpha}-\frac{\Phi_\alpha(1)}{\alpha}
 &=\frac{e^{-\alpha P'(z)}}{\alpha}-H(z)-\frac{e^{-\alpha P'(1)}}{\alpha}+H(1)\\
 &<\frac1\alpha-H(1-\varepsilon),
\end{align*}
where the inequality uses $e^{-\alpha P'(1)}/\alpha>0$, and $H(1)=0$. Since $H(1-\varepsilon)>H(1)=0$, the right side is at most zero whenever $\alpha\geq1/H(1-\varepsilon)$. For such $\alpha$, every $z\leq1-\varepsilon$ satisfies $\Phi_\alpha(z)<\Phi_\alpha(1)$, so no $z\leq1-\varepsilon$ maximizes $\Phi_\alpha$, that is, $\zeta(\alpha)>1-\varepsilon$. Since $\varepsilon$ is arbitrary, $\zeta(\alpha)\to1$ as $\alpha\to\infty$.
\end{proof}

\begin{proof}[Proof of Lemma~\ref{lem:normalization}]
To prove~\eqref{eq:normalization}, we split the levels $z\in(0,1)$ into two
kinds, according to how $\zeta$ reaches them.
Since $\zeta$ is nondecreasing, it has at most countably many discontinuity points. At a discontinuity point $\alpha_0$, let
$z^-=\lim_{\alpha\uparrow\alpha_0}\zeta(\alpha)$, which exists since $\zeta$ is nondecreasing, and $z^+=\zeta(\alpha_0)$, so that $z^-<z^+$ by right-continuity. We call $(z^-,z^+)$ a \emph{jump interval} (note that distinct jump
intervals are disjoint since $\zeta$ is nondecreasing). 
Let $\mathcal{C}$ be the set of levels $z\in(0,1)$ that $\zeta$~attains. For $z\in\mathcal{C}$, we have ${\zeta(\sigma(z))=z}$. Indeed, if $z=\zeta(\alpha_1)$ for some $\alpha_1>0$, then $\sigma(z)\leq\alpha_1$ by the definition of $\sigma(z)$, so $\zeta(\sigma(z))\leq\zeta(\alpha_1)=z$ since $\zeta$ is nondecreasing, while $\zeta(\sigma(z))\geq z$ by right-continuity.
Every  $z\in(0,1)\setminus\mathcal{C}$ lies in a jump interval or is the left endpoint of one. To see this, note that $\zeta(\sigma(z))\geq z$ as shown above, and $\zeta(\sigma(z))\neq z$ since $z$ is not attained, so $\zeta(\sigma(z))>z$. On the other hand, $\zeta(\alpha)<z$ for all $\alpha<\sigma(z)$, so $\lim_{\alpha\uparrow\sigma(z)}\zeta(\alpha)\leq z$. Hence, $\sigma(z)$ is a discontinuity point, and $z$ lies in $[z^-,z^+)$ for its jump interval $(z^-,z^+)$. 

Since there are at most countably many discontinuity points, the set of endpoints of jump intervals is countable, so it has measure zero and does not affect the integrals over $z$ below. Moreover, the levels in a jump interval are not attained by $\zeta$, so $\mathcal{C}$ is disjoint from all jump intervals. Hence, up to a set of measure zero, $(0,1)$ is the disjoint union of $\mathcal{C}$ and the jump intervals.

Let $z\in\mathcal{C}$. Since $\zeta(\alpha)$ maximizes $\Phi_\alpha$ for every $\alpha>0$, the level $z=\zeta(\sigma(z))$ maximizes $\Phi_{\sigma(z)}$, and since $0<z<1$, it is an interior maximizer.
Applying the condition $\alpha P'(z)=-\log(1-z)$ for interior maximizers with $\alpha=\sigma(z)>0$ gives
\begin{equation}\label{eq:regular}
 \sigma(z)P'(z)=-\log(1-z).
\end{equation}

Now consider a jump interval $(z^-,z^+)$ at a discontinuity point $\alpha_0$. We first show that $\sigma(z)=\alpha_0$ for every $z\in(z^-,z^+)$. For $\alpha<\alpha_0$, monotonicity gives $\zeta(\alpha)\leq z^-<z$, while $\zeta(\alpha_0)=z^+>z$. Hence, $\alpha_0$ is the smallest $\alpha$ with $\zeta(\alpha)\geq z$, that is, $\sigma(z)=\alpha_0$.

Next, we show that both endpoints maximize $\Phi_{\alpha_0}$. By definition,  $z^+=\zeta(\alpha_0)$  maximizes $\Phi_{\alpha_0}$. For $z^-$, fix $y\in[0,1]$. For every $\alpha<\alpha_0$, we
have $\Phi_\alpha(\zeta(\alpha))\geq\Phi_\alpha(y)$ since  $\zeta(\alpha)$ maximizes $\Phi_\alpha$. Letting $\alpha\uparrow\alpha_0$ and using the joint continuity of $\Phi_\alpha(z)$ in $(\alpha,z)$, as in the proof of Lemma~\ref{lem:zeta}, gives
$\Phi_{\alpha_0}(z^-)\geq\Phi_{\alpha_0}(y)$. Since every maximizer $z$ of $\Phi_{\alpha_0}$ satisfies $z<1$ and $1-z=e^{-\alpha_0P'(z)}$, both endpoints
satisfy
\begin{equation}\label{eq:endpoints}
 \alpha_0P'(z)=-\log(1-z),\qquad z\in\{z^-,z^+\}.
\end{equation}

Finally, we compare the values of $\Phi_{\alpha_0}$ at the two endpoints. For $z\in[0,1)$, let
\[
 \Psi(z)=z+(1-z)\log(1-z)-\alpha_0P(z).
\]
At each $z\in\{z^-,z^+\}$, substituting $e^{-\alpha_0P'(z)}=1-z$ and $\alpha_0P'(z)=-\log(1-z)$ into the definition of $\Phi_{\alpha_0}$ gives
\begin{align*}
 \Phi_{\alpha_0}(z)
 &=e^{-\alpha_0P'(z)}-\alpha_0\bigl(1-P(z)-(1-z)P'(z)\bigr)\\
 &=(1-z)-\alpha_0+\alpha_0P(z)-(1-z)\log(1-z)\\
 &=1-\alpha_0-\Psi(z).
\end{align*}
Since both endpoints maximize $\Phi_{\alpha_0}$, we have
${\Phi_{\alpha_0}(z^-)=\Phi_{\alpha_0}(z^+)}$, and therefore
$\Psi(z^-)=\Psi(z^+)$. Moreover,
\[
 \Psi'(z) =-\log(1-z)-\alpha_0P'(z),
\]
and $\Psi$ is continuously differentiable on $[z^-,z^+]\subset[0,1)$. Hence, by the fundamental theorem of calculus,
\[
0=\Psi(z^+)-\Psi(z^-)=\int_{z^-}^{z^+}\Psi'(z)\,dz,
\]
that is,
\begin{equation}\label{eq:jump}
 \int_{z^-}^{z^+}-\log(1-z)\,dz=\alpha_0\int_{z^-}^{z^+}P'(z)\,dz.
\end{equation}

We now combine~\eqref{eq:regular} and~\eqref{eq:jump}. For $z\in\mathcal{C}$, \eqref{eq:regular} gives $-\log(1-z)/\sigma(z)=P'(z)$. On a jump interval $(z^-,z^+)$ at $\alpha_0$, we have $\sigma(z)=\alpha_0$, so dividing \eqref{eq:jump} by $\alpha_0$ gives $\int_{z^-}^{z^+}-\log(1-z)/\sigma(z)\,dz=\int_{z^-}^{z^+}P'(z)\,dz$.
Adding the contributions of $\mathcal{C}$ and of all jump intervals, which is justified since the integrands are nonnegative, and using $P(0)=0$ and $P(1)=1$, we obtain
\[
\int_0^1\frac{-\log(1-z)}{\sigma(z)}\,dz=\int_0^1P'(z)\,dz=P(1)-P(0)=1.
\]
\end{proof}


\newpage
\begin{thebibliography}{30}

\bibitem{cover}
D. Bar-Lev, O. Sabary, R. Gabrys, and E. Yaakobi,
``Cover your bases: How to minimize the sequencing coverage in DNA storage
systems,'' \emph{IEEE Trans. Inf. Theory}, vol.~71, no.~1, pp.~192--218, 2025.

\bibitem{duality}
M. Bertuzzo, A. Ravagnani, and E. Yaakobi, ``The DNA coverage depth problem:
Duality, weight distributions, and applications,'' arXiv:2603.06489, 2026.

\bibitem{optimal}
R. Gross, Y. Grunbaum, M. Bertuzzo, E. Yaakobi, and A. Ravagnani,
``Optimal codes for the coverage depth problem and the performance of random codes,''
arXiv:2610.07370, 2026.

\bibitem{files}
H. Abraham, R. Gabrys, and E. Yaakobi,
``Covering all bases: The next inning in DNA sequencing efficiency,''
in \emph{Proc. IEEE Int. Symp. Inf. Theory (ISIT)}, 2024, pp.~464--469.

\bibitem{generalized}
A. Gruica, A. Petrillo, and F. Zullo,
``The generalized random access problem for linear codes,''
arXiv:2608.20152, 2026.

\bibitem{block}
D. Bar-Lev,
``Coded information retrieval for block-structured DNA-based data storage,''
arXiv:2603.17154, 2026.

\bibitem{hyperbolic}
C.~V.~A. Vlachos and D. Bar-Lev,
``A hyperbolic bound for file retrieval in DNA-based data storage,''
arXiv:2609.36067, 2026.

\bibitem{partial}
N. Akkaya and E. Haytao\u{g}lu,
``Dupleth: DNA depolama sistemlerinde LT kodlar{\i} kullan{\i}ld{\i}\u{g}{\i}nda k{\i}smi kapsama derinli\u{g}i problemi,''
\emph{Osmaniye Korkut Ata Üniversitesi Fen Bilimleri Enstitüsü Dergisi},
vol.~9, no.~4, pp.~2301--2321, 2026.

\bibitem{hanna}
S.~K. Hanna,
``On the reliability of information retrieval from MDS coded data in DNA storage,''
\emph{IEEE Trans. Mol. Biol. Multi-Scale Commun.}, 2026.

\bibitem{cao}
R. Cao, P. Zhou, and X. Chen,
``Modeling DNA storage retrieval reliability via sequencing coverage depth,''
\emph{Brief. Bioinform.}, vol.~27, no.~3, Art.~no.~bbag303, 2026.

\bibitem{noisy}
H. Abraham, I. Feldman, and E. Yaakobi,
``Coverage depth analysis: Random access for DNA storage over noisy channels,''
in \emph{Proc. IEEE Int. Symp. Inf. Theory (ISIT)}, 2026, pp.~1--6.

\bibitem{repair}
A. Levy, R. Con, E. Yaakobi, and H.~M. Kiah,
``Expected recovery time in DNA-based distributed storage systems,''
arXiv:2602.07601, 2026.

\bibitem{distributed}
X. Kong, O. Elishco, C. Wang, and T.~M. Duman,
``The coverage depth problem in distributed DNA data storage,''
arXiv:2610.02931, 2026.

\bibitem{windows}
Y. Grunbaum and E. Yaakobi,
``General coverage models: Structure, monotonicity, and shotgun sequencing,''
in \emph{Proc. IEEE Int. Symp. Inf. Theory (ISIT)}, 2026, pp.~1--6.

\bibitem{preuss}
I. Preuss, B. Galili, Z. Yakhini, and L. Anavy,
``Sequencing coverage analysis for combinatorial DNA-based storage systems,''
\emph{IEEE Trans. Mol. Biol. Multi-Scale Commun.},
vol.~10, no.~2, pp.~297--316, 2024.

\bibitem{composite}
T. Cohen and E. Yaakobi,
``Optimizing the decoding probability and coverage ratio of composite DNA,''
\emph{IEEE J. Sel. Areas Inf. Theory}, vol.~6, pp.~417--431, 2025.

\bibitem{motifs}
R. Sokolovskii, P. Agarwal, L.~A. Croquevielle, Z. Zhou, and T. Heinis,
``Coding over coupon collector channels for combinatorial motif-based DNA storage,''
\emph{IEEE Trans. Commun.}, vol.~73, no.~6, pp.~3750--3760, 2025.

\bibitem{combinatorial}
A. Gruica, D. Bar-Lev, A. Ravagnani, and E. Yaakobi,
``A combinatorial perspective on random access efficiency for DNA storage,''
\emph{IEEE Trans. Inf. Theory}, vol.~71, no.~12, pp.~9395--9412, 2025.

\bibitem{geometry}
A. Gruica, M. Montanucci, and F. Zullo,
``The geometry of codes for random access in DNA storage,''
\emph{Des. Codes Cryptogr.}, vol.~94, no.~5, Art.~no.~114, 2026.

\bibitem{variables}
\c{S}. Bodur, S. Lia, H.~H. L\'opez, R. Ludhani, A. Ravagnani, and L. Seccia,
``The random variables of the DNA coverage depth problem,''
\emph{IEEE Trans. Inf. Theory}, 2026, doi:~10.1109/TIT.2026.3711141.

\bibitem{wang}
C. Wang and E. Yaakobi, ``Random access in DNA storage: Algorithms,
constructions, and bounds,'' arXiv:2601.07053, 2026.


\bibitem{matroid}
D. Bar-Lev, H.~H. López, G.~L. Matthews, and F. Salizzoni,
``A matroidal framework for random access in linear codes,''
arXiv:2610.28820, 2026.

\bibitem{first}
A. Boruchovsky, O. Elishco, R. Gabrys, A. Gruica, I. Tamo, and E. Yaakobi,
``Making it to first: The random access problem in DNA storage,''
arXiv:2501.12274, 2025.

\bibitem{hby}
C. Hofmeister, R. Bitar, and E. Yaakobi,
``Random access expectation in DNA storage and fountain codes,''
arXiv:2605.10919v2, 2026.

\bibitem{rank}
A. Coja-Oghlan, A.~A. Erg\"ur, P. Gao, S. Hetterich, and M. Rolvien,
``The rank of sparse random matrices,'' \emph{Random Struct. Algorithms},
vol.~62, no.~1, pp.~68--130, 2023.
\end{thebibliography}
\end{document}